\documentclass[conference]{IEEEtran} 
\usepackage[usenames,dvipsnames]{xcolor}

\usepackage{amsmath}
\usepackage{amssymb}
\usepackage{overpic}
\usepackage{framed}
\usepackage{multirow}
\usepackage{hhline}
\usepackage{lettrine}
\usepackage{bm}
\usepackage{graphicx}
\usepackage{caption}
\usepackage{algorithm}
\usepackage{algorithmic}
\newtheorem{theorem}{Theorem}

\newtheorem{lemma}[theorem]{Lemma}

\newtheorem{definition}{Definition}

\newcommand{\ve}{\varepsilon}
\newcommand{\tcb}[1]{}                                             % to remove all the in-line comments

\makeatletter
\def\moverlay{\mathpalette\mov@rlay}
\def\mov@rlay#1#2{\leavevmode\vtop{%
   \baselineskip\z@skip \lineskiplimit-\maxdimen
   \ialign{\hfil$\m@th#1##$\hfil\cr#2\crcr}}}
\newcommand{\charfusion}[3][\mathord]{
    #1{\ifx#1\mathop\vphantom{#2}\fi
        \mathpalette\mov@rlay{#2\cr#3}
      }
    \ifx#1\mathop\expandafter\displaylimits\fi}
\makeatother

\newcommand{\bigcupdot}{\charfusion[\mathop]{\bigcup}{\cdot}}

\begin{document}
\title{Scheduling Algorithms for 5G Networks with Mid-haul Capacity Constraints}
\author{
\IEEEauthorblockN{Abhishek Sinha\IEEEauthorrefmark{1}, Matthew Andrews\IEEEauthorrefmark{2}, Prasanth Ananth\IEEEauthorrefmark{2}}
\IEEEauthorblockA{\IEEEauthorrefmark{1}Dept. of Electrical Engineering, IIT Madras, Chennai, India}
\IEEEauthorblockA{\IEEEauthorrefmark{2}Nokia Bell Labs, Murray Hill, NJ, USA}
Email: \IEEEauthorrefmark{1} abhishek.sinha@ee.iitm.ac.in,
\IEEEauthorrefmark{2} matthew.andrews@nokia-bell-labs.com,
\IEEEauthorrefmark{2} prasanth.ananth@nokia-bell-labs.com
}
\maketitle
\begin{abstract}
We consider a virtualized RAN architecture for 5G networks where the Remote Units are connected to a central unit via a mid-haul. To support high data rates, the mid-haul is realized with a Passive Optical Network (PON). In
this architecture, the data are stored at the central unit until the
scheduler decides to transmit it through the mid-haul to an appropriate remote unit, and then over the air at the same slot. 
We study an optimal scheduling problem that arises in this context. This problem has two key features. First, multiple cells must be scheduled simultaneously for efficient operation. Second, the interplay between the time-varying wireless interface rates and the fixed capacity PON needs to be handled efficiently.  In this paper, we take a comprehensive look at this resource allocation problem by formulating it as a utility-maximization problem. Using combinatorial techniques, we derive useful structural properties of the optimal allocation and utilize these results to design polynomial-time approximation algorithms and a pseudo-polynomial-time optimal algorithm. Finally, we numerically
compare the performance of the proposed algorithms to heuristics which are natural generalizations of the ubiquitous Proportional Fair algorithm. 
\end{abstract}

%\section{Introduction}

\section{Introduction}
\lettrine[]{\textbf{T}}{ }wo inexorable trends will have a significant impact on the future of 5G wireless access. The first is the trend towards denser small cells
with deeper fiber, which is sometimes described using the slogan
``long wires and short wireless''. The second is the trend towards
virtualized Radio Access Network (vRAN) architectures in which part of the
processing and network intelligence (including scheduling) takes place in the central units
(CUs) located in a cloud data center (sometimes called an {\em edge cloud}) and 
then the data are carried over a transport network called {\em mid-haul} to 
a set of remote units (RUs).

A passive optical network (PON) is ideally suited for such mid-haul due to
its high capacity, lower cost, and ability to reuse the existing fiber-to-the-x
(FTTx) distribution networks. However, if we utilize such an
architecture, then we need to ensure that the scheduling decisions in the central
units respect the limited PON capacity, in addition to the time-varying air interface data rate. The goal of this paper is to investigate how such scheduling can be carried out efficiently.

There are many variants of the vRAN architecture that differ based on how the
processing is split between the CUs and the RUs. 
At a high-level, we can categorize these options into two types. 
In a {\em front-haul} architecture, all processing right down to
the baseband takes place in the edge cloud.  On the contrary, in this paper, we will be concerned with the so-called {\em mid-haul} architecture \cite{pfeiffer2015next}, \cite{DotschDMSSS13}, where some of the higher-layer processing takes place in the edge cloud 
while the lower physical layer processing takes place at the RUs. Hence, the mid-haul architecture requires less PON bandwidth compared to the front-haul architecture. However, 
the mid-haul bandwidth requirement changes with time depending on the actual amount of 
user traffic and the instantaneous wireless channel conditions. 
In this paper, we address the following question- How should the
central units schedule the wireless transmissions at RUs efficiently in the \emph{full-buffer} traffic regime so that the \emph{fixed} PON capacity constraint and the \emph{time-varying} wireless interface rate constraints are satisfied (see Figure~\ref{PON-fig})?
 
Intuitively, the centralized scheduler must take into account the limited PON capacity in addition to the instantaneous air interface channel conditions for efficient operation. To
minimize latency (\emph{e.g.,} in the case of URLLC traffic), the scheduling should be done in such a way that there is \emph{no queue build-up} at the RUs.

The conflicting nature of the constraints makes this resource allocation problem challenging to solve.
For the wireless air interface, the fundamental
resource units are the {\em resource blocks} (RBs), which give rise to time-varying bit
rates according to the dynamic wireless channel conditions. On the
other hand, for the PON, the fundamental resource units are \emph{fixed-capacity} PON slices. As a concrete example, the air interface scheduler may 
wish to serve a user that is in a good channel condition. However, 
it may not be able to do that if the PON cannot
handle the resulting data rate. In this paper, we undertake a comprehensive study of this problem and develop algorithms with provable guarantees to solve it efficiently. \\
The rest of the paper is organized as follows:
\begin{itemize}
\item  In Section~\ref{formulation}, we formulate the problem by using the theory of gradient ascent over time-varying channels \cite{stolyar2005asymptotic}. This methodology
  decomposes the long-run average utility objective into slot-by-slot local objectives.
\item  In Section~\ref{subopt}, we give an illustrative example to show why
  the standard greedy algorithms (including, \emph{e.g.}, the Proportional Fair scheduling) fail to provide an optimal
  solution. The fundamental difficulty is that greedy approaches
  cannot optimally handle the mismatch between the air interface constraint and the
  PON capacity constraint.
% We also prove that the corresponding single-slot problem is NP-hard. 
\item  In Section~\ref{structure}, we present an efficient algorithm that gives the optimal wireless rate allocations for a fixed RB assignment to the user. This algorithm is used in our later developments. 
%
%$\bullet$ In Section~\ref{s:prelim}, we present a number of preliminary results,
%  including two natural heuristics, a
%  structural result for the optimal solution, and two linear
%  programming relaxations of the original problem.\\
\item In Section~\ref{s:singleC}, we present two algorithms for the special
  case when only the overall PON capacity constraint is binding. In
  particular, we present a polynomial-time $2$-approximation algorithm
  based on LP-rounding. This algorithm exploits the
  special structure of the basic feasible solutions of the associated LP. We
  also provide an optimal Dynamic Programming (DP) algorithm that runs in
  pseudo-polynomial time. 
  \item In Section~\ref{s:algo}, we present a greedy $2$-approximation 
  matroid-based algorithm for the general case, where the individual RU-specific capacity constraints, as well as the overall PON capacity constraint are active. 
 \item  In Section~\ref{s:prelim}, we present two natural heuristics which are inspired from the ubiquitous Proportional Fair algorithm. 
\item  In Section~\ref{simulation}, we evaluate the proposed algorithms via numerical simulation and examine how they compare to the heuristics. 
\item  Finally, Section \ref{discussions} concludes the paper after a brief discussion on related works.
\end{itemize}
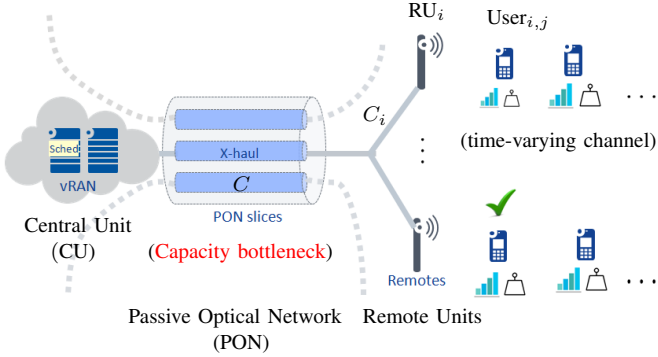
\begin{figure}
	\centering
		\begin{overpic}[width=0.47\textwidth]{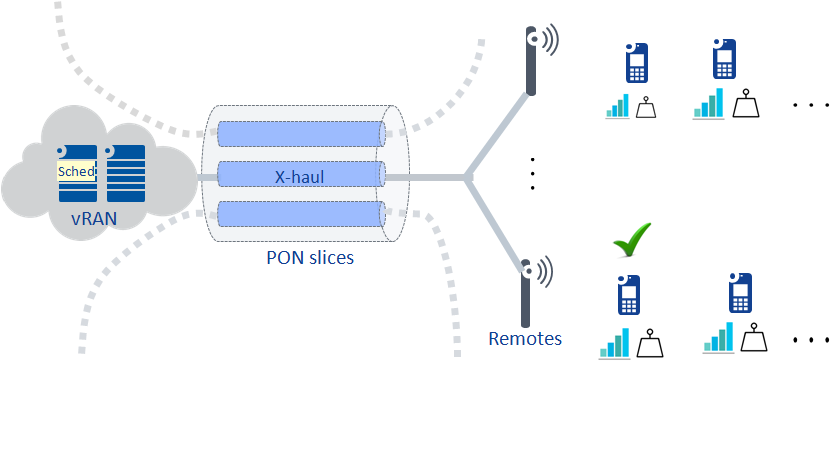}
		\put(55,39){\footnotesize{$C_i$} }
		\put(62,55){\footnotesize{$\textrm{RU}_i$} }
		\put(7,18){\footnotesize{$(\textrm{CU}$)} }
		\put(74,54){\footnotesize{$\textrm{User}_{i,j}$} }
		\put(35,28.3){\footnotesize{$C$} }
		\put(22,18){\footnotesize{(\textcolor{red}{Capacity bottleneck})}}
		\put(3,22){\footnotesize{\textcolor{black}{Central Unit}}}
		\put(55,8){\footnotesize{\textcolor{black}{Remote Units}}}
		\put(19,8){\footnotesize{\textcolor{black}{Passive Optical Network}}}
		\put(32,4){\footnotesize{\textcolor{black}{(PON)}}}
		\put(70,35){\footnotesize{(time-varying channel)}}
		%\end{overpic}	
		\end{overpic}
        \caption{Schematic of a split-processing vRAN architecture}
		\label{PON-fig}
\end{figure}
%\vspace{-40pt}
\section{Problem Formulation} \label{formulation}
\emph{System Model:} We consider a split-processing vRAN architecture as shown in Figure~\ref{PON-fig}.
There are $m$ Remote Units (RUs) that transmit data over the air to the  wireless end
users connected to it. A Passive Optical Network (PON) connects a Central Unit (CU)
to the Remote Units (RUs). At each time slot, each RU can transmit data on $\kappa$ Resource
Blocks (RBs), each of which corresponds to a set of contiguous OFDM
carriers. At every RU, each RB can be assigned to at most one user at a slot. We assume that inter-RU interference is negligible, which is a reasonable assumption given the advent of the CoMP technology \cite{liu2012combating}.  

\emph{Traffic Model:}  
For simplicity, we focus on the downlink
traffic only. We assume that the users are infinitely backlogged, i.e., the users' respective data buffers in the CU are always full. In order to keep the latency small, data is \emph{not buffered} in the Remote Units. Hence, all scheduled data from the CU must be delivered to the users through the PON and the wireless interface at the RUs at the \emph{same} slot. 
All scheduling decisions are assumed to be made by the CU. 

\emph{Decision Variables and Constraints:} We use the symbol $i$ to index the RUs, the pair $(i,j)$ to index the $j$\textsuperscript{th} user associated with the $i$\textsuperscript{th} RU (denoted by $\textrm{RU}_i$), and
$k$ to index the RBs. Let $C$ be the total capacity of the PON mid-haul, and let $C_i$ be the capacity of the optical fibers connecting $\textrm{RU}_i$ to the PON (see Fig. \ref{PON-fig}).
Let $\gamma_{ijk}(t)$ denote
the instantaneous air-interface rate for the user $(i,j)$ on the RB $k$ at slot $t$. 
Let $x_{ijk}(t) \in \{1,0\}$ be a binary \emph{decision variable}
representing whether or not the RB $k$ at $\textrm{RU}_i$ is
assigned to the user $(i,j)$ at slot $t$. Let $y_{ijk}(t)$ be a non-negative \emph{decision variable} denoting the rate
allocated to the user $(i,j)$ on the RB $k$ at slot $t$. 

The limited PON capacity and air-interface rates enforce the following
constraints on the instantaneous decision variables $\bm{x}(t),\bm{y}(t)$: 
\begin{eqnarray}
y_{ijk}(t)&\le& \gamma_{ijk}(t)x_{ijk}(t),~~~\forall i,j,k,\label{eq:rate}\\
\sum_j x_{ijk}(t)&\le& 1,~~~~\forall i,k,\label{eq:uniq}\\
\sum_{jk} y_{ijk}(t)&\le& C_i,~~~~\forall i,\label{eq:base}\\
\sum_{ijk} y_{ijk}(t)&\le& C,\label{eq:tot}\\
x_{ijk}(t)&\in&\{0,1\}, ~ y_{ijk}\geq 0, ~~~\forall i,j,k \label{eq:bin}.
\end{eqnarray}
\emph{Discussions on the constraints \eqref{eq:rate}-\eqref{eq:bin}:} The inequality~(\ref{eq:rate}) reflects the fact that the allocated rate $y_{ijk}(t)$ can be
non-zero only if the RB $k$ is assigned to the user $(i,j)$ at time
$t$. Moreover, due to time-varying nature of the air-interface rate, the allocated data rate at slot $t$ can be at most
$\gamma_{ijk}(t)$ for that assignment. Inequality~(\ref{eq:uniq}) states that at most one user may be assigned to any RB $k$ on $\textrm{RU}_i$. Inequalities
(\ref{eq:base}) and (\ref{eq:tot}) denote the mid-haul capacity
constraints. The constraint \eqref{eq:bin} simply denotes the fact that the variables $x_{i,j,k}(t)$ are binary and the allocated rates are non-negative.

\emph{Special Case:} In general, the capacity of optical fibers decrease sharply with their length \cite{ramaswami2009optical}.  Since the distance between the PON remote end-point and any RU is much shorter than the size of the PON itself, it is often the case that $C_i >> C, \forall i$.  To exploit this fact in designing algorithms, we pay particular attention to the important special case where the RU specific constraints \eqref{eq:base} are relaxed (effectively by setting $C_i=\infty, \forall i $), and the system is limited by the overall PON capacity constraint \eqref{eq:tot} only (see Section \ref{s:singleC}). 

\emph{Objective:} In this paper, we formulate the resource allocation problem as a utility maximization problem. Let $\bar{r}_{ij}$ be the long-term average data rate for the user $(i,j)$, i.e., 
\begin{equation*}
	\bar{r}_{ij}= \liminf_{T \to \infty} \frac{1}{T}\sum_{t=1}^{T} \sum_{k} y_{ijk}(t).
\end{equation*}
Let $U(\cdot)$ be a strictly concave smooth utility function. Our objective is to find a scheduling policy that maximizes the sum-utility of all users defined as follows:
\begin{equation} \label{util}
U(\bar{\bm{r}})=\sum_{ij} U(\bar{r}_{ij}). 	
\end{equation}
Following the development in \cite{stolyar2005asymptotic}, for each user $(i,j)$, we first define an exponentially smoothed long-term service rate $R_{ij}(t)$, which
 evolves as follows:
\begin{eqnarray} \label{state-update}
	R_{ij}(t+1)=(1-\beta)R_{ij}(t)+\beta\sum_k y_{ijk}(t), ~~ R_{ij}(1)=0,
\end{eqnarray}
 where $\beta > 0$ is a small positive constant. From \cite{stolyar2005asymptotic}, it follows that the long-term objective \eqref{util} is maximized by
finding the instantaneous decision variables $\bm{x}(t), \bm{y}(t)$ at each slot $t$ that solves the following problem, referred to as {\sc Single Shot}: 
\begin{eqnarray} \label{opt1} 
\max_{\bm{x}(t), \bm{y}(t)}\sum_{ij}U'(R_{ij}(t))\sum_k y_{ijk}(t),
\end{eqnarray}
subject to the constraints (\ref{eq:rate})-(\ref{eq:bin}). We
therefore focus on the {\sc Single-Shot} problem for the slot $t$ for the remainder of
the paper. 

For the common case of logarithmic utility, in which $U(x)=\log x$, the above per-slot problem \eqref{opt1} becomes:
\begin{eqnarray} \label{opt}
\max_{\bm{x}(t), \bm{y}(t)} \sum_{ij}\frac{\sum_k y_{ijk}(t)}{R_{ij}(t)},
\end{eqnarray}
subject to the constraints (\ref{eq:rate})-(\ref{eq:bin}).
For concreteness, we will use the objective (\ref{opt}) throughout the paper,
but all our results apply to any strictly concave smooth utility
function $U(\cdot)$.

 It should be noted that in the case of inelastic traffic, where the objective is to ensure network stability while achieving throughput-optimality, a standard algorithm is {\sc Max-Weight} \cite{tassiulas1992stability}. In the {\sc Max-Weight} algorithm, one is required to solve an identical problem to {\sc Single Shot} at every slot, where  the factor $R_{ij}(t)^{-1}$ in the objective \eqref{opt} is replaced with the corresponding queue-length. Hence, the algorithmic techniques that we develop for {\sc Single Shot} directly apply to the case of {\sc Max-Weight} algorithm for inelastic traffic. \\
 Since, we solve the {\sc Single Shot} problem at every slot $t$, to avoid notational clutter, we will be dropping the time-argument $t$ from all variables henceforth. 
\section{Sub-Optimality of the Proportional Fair Scheduler}  \label{subopt}
%\section{Suboptimality of Proportional Fair} \label{subopt}

For solving the {\sc Single Shot} problem \eqref{opt}, the well-known Proportional Fair scheduler ({\sc PF}) assigns the RB $k$ at $\textrm{RU}_i$ 
to the user $(i,j)$ that maximizes the index $\gamma_{ijk}/R_{ij}$ \cite{PF}.  It does so without taking into account the capacity constraints (\ref{eq:base})-(\ref{eq:tot}) for the PON. 
We prove that the {\sc PF} algorithm is not optimal 
due to the presence of the capacity constraints. 
%Let $X^*(t)$ be
%the optimal value of the aggregate utility \eqref{opt} achievable at time $t$. 
%This implies that the performance
%of the {\sc PF} algorithm can be an arbitrarily large factor worse than the optimal performance.
\begin{framed}
\begin{lemma}[Sub-Optimality of PF]
\label{l:greedynotoptimal}
The Proportional Fair (PF) scheduler is not optimal for the {\sc Single Shot} problem \eqref{opt}, in general.
% in general, i.e.,\ there exist values of
%$\gamma_{ijk}(t),R_{ij}(t),C_i,C$ such that for the $y_{ijk}(t)$ produced by
%PF we have,
%$$
%\sum_{ij}\frac{\sum_k y_{ijk}(t)}{R_{ij}(t)}<X^*(t).
%$$
\end{lemma}
\end{framed}
\begin{IEEEproof}
Our counter-example has one RU with two users and four RBs. 
Let $C=7$. Note that, we don't need to specify the separate $C_i$ values since we
have only one RU. Let the aggregate rates for the two users at slot $t$ be,
$$
R_{00}=1~~R_{01}=2.
$$
Let the instantaneous channel rates be,
$$\gamma_{00k}=1~~~
\gamma_{01k}=4~~~\forall k. 
$$
Since $\gamma_{01k}/R_{01}=2>1=\gamma_{00k}/R_{00}$,
PF will pick user 1 for every RB,
i.e.\ $x_{01k}=1$ and $x_{00k}=0$ for all $k$. Given the PON capacity constraint, we choose the $y$ values such that $\sum_k
y_{01k}=7$ and $\sum_k y_{00k}=0$. Hence, the total objective value for {\sc Single Shot} is
$7/2$. 

A better solution would put user $0$ on 3 RBs and user $1$
on a single RB. In this case we have $x_{00k}=1$ and
$y_{00k}=1$ for $k<3$ and $x_{013}=1$ and $y_{013}=4$. The total
objective for this solution is $\frac{3}{1}+\frac{4}{2}=5$. 
\end{IEEEproof}
%\section{Heuristic Algorithms}
%\label{s:prelim}
%\subsection{Heuristics}
%
%We now present two heuristics that provide a baseline for comparison
%with our proposed algorithms appearing later in Section \ref{s:singleC} and \ref{s:algo}.\\~\\
%$\bullet$ {\sc max-yield}. This is the simplest adaptation of the traditional 
%PF algorithm so that it respects the capacity constraints. The algorithm works by 
%going through the RUs and RBs in decreasing order of the index $ \max_j \gamma_{ijk}/R_{ij}$ sequentially and always picking the user
%$j$ that maximizes the index $\gamma_{ijk}/R_{ij}$. 
%At each step, the used and remaining capacity on the PON is tracked and the algorithm 
%stops when the available capacity is exhausted.\\~\\
%%subject to there being available capacity on the PON. 
%$\bullet$ {\sc max-value}. This algorithm works by going through
%the RUs and RBs in decreasing order of the index $ \max_j \gamma_{ijk}/R_{ij}$ sequentially and always picking the user that maximizes $1/R_{ij}$.
%At each step, the used and remaining capacity on the PON are tracked and the algorithm 
%stops when the available capacity is exhausted. 
%Note that {\sc max-yield} tries to optimize the objective with respect to the wireless 
%resources, and algorithm {\sc max-value} tries to optimize the objective with respect to 
%the PON capacity constraints.
%
\section{Structural Results}
\label{structure}
In some sense, the difficulty of maximizing \eqref{opt} stems from the fact that
the optimal solution may split the total RB allocation between the users with high values
of $1/R_{ij}$ and the users with high values of $\gamma_{ijk}/R_{ij}$. Recall from
Lemma~\ref{l:greedynotoptimal} that, if the {\sc PF} algorithm violates
the capacity constraints, then it might lead to a suboptimal allocation. We now state an intuitively obvious result that if this violation does not happen then, {\sc PF} is, in fact, optimal. 
 
\begin{lemma}
Suppose that with the {\sc PF} allocation, we can set
$y_{ijk}=\gamma_{ijk}x_{ijk}, \forall i,j,k,$ without violating the
capacity constraints \eqref{eq:base}-\eqref{eq:tot}. Then {\sc PF} achieves optimality. 
\end{lemma}
\begin{IEEEproof}
Follows from the observation that the maximum objective value that we can obtain from
RB $k$ at $\textrm{RU}_i$ is $\max_j \gamma_{ijk}/R_{ij}$. If the {\sc PF} Algorithm achieves this value without violating the capacity constraints, then it
is optimal. 
\end{IEEEproof}

In order to proceed further, Lemma \ref{l:structure} presents a simple and efficient strategy that
allows us to determine the optimal allocation $\bm{y}_{\textrm{opt}}$ for any given feasible assignment $\bm{x}$. As a consequence of Lemma \ref{l:structure}, the problem \eqref{opt} reduces to determining the optimal assignment $\bm{x}^*$, \emph{i.e.}, the 
identity of the user $j$ to which the RB $k$ on $\textrm{RU}_i$ should be
assigned. The optimal amount of service $y_{ijk}$ that the user $(i,j)$ then receives is
determined by this lemma. 
\begin{lemma}
\label{l:structure}
Suppose that we are given a set of binary $\bm{x}$ values that satisfy the
feasibility constraint \eqref{eq:uniq}. We can find the
optimal rate allocation $\bm{y}_{\textrm{opt}}$ for this set of $\bm{x}$ via the following iterative scheme: 
\begin{algorithm}
\caption{Optimal Rate Allocation ($\bm{y}_{\textrm{opt}}$) for a given RB Assignment Profile ($\bm{x}$)}
\label{findy}
\begin{algorithmic}[1]
\STATE Set $\bm{y}\gets \bm{0}$.
\STATE Order the $(ijk)$ triples in decreasing order of the value of $1/R_{ij}$.  
\STATE Go
through each triple sequentially in order. When considering the triple $(ijk)$, set
$$
y_{ijk}\leftarrow \min\{\gamma_{ijk}x_{ijk},C-\sum_{i'j'k'}y_{i'j'k'},C_i-\sum_{j'k'}y_{ij'k'}\}.
$$
%In other words, set $y_{ijk}(t)$ as large as possible without
%violating any of the constraints. 
\STATE Set $\bm{y}_{\textrm{opt}} \gets \bm{y}$.
\end{algorithmic}
\end{algorithm}
\end{lemma}
\begin{IEEEproof}
See Appendix \ref{structure:proof}.
\end{IEEEproof}

\section{Algorithms for {\sc Single Shot} with an overall PON Capacity Constraint}
\label{s:singleC}
In this Section, we consider the special case where the separate RU-specific capacity constraints \eqref{eq:base} are relaxed (by effectively setting $C_i=\infty, \forall i$). Thus the system is capacity-limited by the PON constraint \eqref{eq:tot} only. As discussed in Section \ref{formulation}, this is a practically relevant case when the PON size is much larger than the RU-to-PON access distances.\\ 
It is not hard to see that, this special case is algorithmically equivalent to the scenario where there is only one RU, and all UEs are associated with this RU. Thus, to avoid notational clutter, we drop the RU index $i$ throughout this Section. We start with the following definition:
\begin{definition}
	A feasible rate allocation vector $\bm{y}$ is called \textsc{Discrete} if $y_{jk}=x_{jk}\gamma_{jk}, \forall k$. 
\end{definition}
In other words, in a \textsc{Discrete} allocation either the RB is
allocated the maximum wireless rate given by the wireless interface rate ($\gamma_{jk}$) or it is not allocated any rate at all.
\begin{definition}
	A feasible rate allocation vector $\bm{y}$ is called \textsc{Almost
          Discrete} if $y_{jk}=x_{jk}\gamma_{jk}$ for all RBs, possibly
        excepting \emph{at most one} RB. 
\end{definition}
The significance of the above definition is borne out by the following Theorem.
\begin{framed}
\begin{theorem}\label{ad}
There exists an  optimal solution to {\sc Single Shot} which is \textsc{Almost Discrete}. 
\end{theorem}
\end{framed}
\begin{IEEEproof}
Follows directly from Lemma~\ref{l:structure} since it implies
that once the $\bm{x}$ values are set, we can find the optimal
$\bm{y}$ values by simply going through each of them in decreasing the order of
$1/R_j$. Excepting the last one, each one is filled up to an amount $\gamma_{jk}$ before
moving on to the next one.  
\end{IEEEproof}
%\iffalse
%Let $\bm{x}^*$ be the binary part of an optimal solution of \eqref{opt}. Now we show that there exists an \textsc{Almost Discrete} feasible optimal allocation $\bm{y}^*$.	\\
%Since $\bm{x}^*$ is the binary portion of a  feasible of \eqref{opt}, without any loss of generality, for each RB $k$ there exists exactly one user $j^*(k)$ such that $x^*_{j^*(k)k}=1$.  Denote $R_{j^*(k)k}=R_k$ and $\gamma_{j^*(k)k}=\gamma_k$. Then substituting this optimal $\bm{x^*}$ in \eqref{opt}, we obtain the following LP for the optimal $\bm{y^*}$: 
%
%\begin{eqnarray*}
%	\max \sum_k y_k/R_k
%\end{eqnarray*}
%
%s.t., 
%\begin{eqnarray}
%	\sum_k y_k &\leq& C \label{simplex1}\\
%	0 \leq y_k &\leq& \gamma_k, \hspace{10pt} \forall k  \label{bd2}
%\end{eqnarray}
%Clearly, the feasible region of the LP above is bounded and hence a finite optimal solution exists. 
%Now we recall the fundamental result that the solution of an LP is always obtained at a vertex of the feasible region, which is also known as the \textbf{B}asic \textbf{F}easible \textbf{S}olution (BFS). \\
%Since there are $(K+1)$ linearly independent inequalities in
%\eqref{simplex1} and \eqref{bd2} and the dimension of the vector
%$\bm{y}$ is $K$, it is clear that at least $(K-1)$ of inequalities in
%Eqn. \eqref{bd2} must be active in some optimal solution. This proves
%that there exists an optimal solution which is \textsc{Almost
%  Discrete}. 
%\fi
\emph{Note:} It is also possible to prove Theorem \ref{ad} without appealing to Lemma \ref{l:structure}. See Appendix \ref{alter} for an alternative proof using combinatorial properties of the associated LP.

In the following, we exploit the result in Theorem \ref{ad} to design a polynomial-time approximation algorithm to the {\sc Single Shot} problem with an overall PON capacity constraint. 
\subsection{Poly-time 2-Approximation Algorithm {\sc Rounding-AD}} \label{LP}
%{An LP-based Algorithm} 
We now design an LP-based algorithm {\sc Rounding-AD} ({\sc AD} stands for {\sc Almost Discrete}) 
for approximately solving the {\sc Single Shot} problem. By substituting $y_{jk}\gets \gamma_{jk}x_{jk}$ and relaxing $x_{jk}$ to take any real number in the interval $[0,1]$, we obtain the following LP relaxation to \eqref{opt}:
\begin{eqnarray} \label{rlp}
	\max \sum_{jk} x_{jk} \frac{\gamma_{jk}}{R_j}
\end{eqnarray}
subject to, 
\begin{eqnarray}
	\sum_{j} x_{jk} &\leq&  1 , \hspace{10pt} \forall k, \label{constr1}\\
	\sum_{jk} \gamma_{jk} x_{jk} &\leq&  C,  \label{cap_constr}\\
	\bm{x} &\geq& 0.
\end{eqnarray}
Call the above relaxed Linear Program $\textsc{RLP}$. The following Theorem shows that an optimal solution to {\sc RLP} is ``close'' to being an all-integral solution.  
\begin{framed}
\begin{theorem} \label{th1}
An optimal solution to $\textsc{RLP}$ allocates every RB to at most one user, excepting, possibly \emph{at most} one RB (which is shared between two users). 	
\end{theorem}
\end{framed}
\begin{IEEEproof}
By introducing the non-negative auxiliary variables $\zeta_k, \forall k$ in \eqref{constr1} and $\xi$ in \eqref{cap_constr}, we obtain the following set of equivalent constraints:

\begin{eqnarray} 
	\sum_j x_{jk} + \zeta_k &=&1,  ~~\forall k, \label{eq_constr1}\\
	\sum_{jk} \gamma_{jk} x_{jk} + \xi &=&  C,  \label{eq_constr2}
\end{eqnarray}
where $\bm{x}, \bm{\zeta}, \xi \geq 0$. 
Next, recall that. for any LP, an optimal solution (also known as a \textbf{B}asic \textbf{F}easible \textbf{S}olution (\textbf{BFS})), is always obtained at some vertex of the polytope defined by the constraints \cite{papadimitriou1998combinatorial}. Let the total number of RBs be $\kappa$. Since there are a total of $(\kappa+1)$ equality constraints taken together in the equality constraints \eqref{eq_constr1}-\eqref{eq_constr2}, at most $(\kappa+1)$ variables could be strictly positive in any {\bf BFS}. Also, since the RHS of the constraints in \eqref{eq_constr1} are positive, it follows that there is \emph{at least one} strictly positive variable per equality constraints \eqref{eq_constr1}. This implies, by the pigeonhole principle, that, in an optimal solution to {\sc RLP}, there could be at most one RB $k_1$ which has been allocated to two users $j_1, j_2, (j_1\neq j_2)$ (i.e., $x_{j_1k_1}x_{j_2k_1}>0$). Moreover, all other RBs have been allocated to at most one user in the optimal solution. 
\end{IEEEproof}

Note that, if the optimal solution to {\sc RLP} contains no fractional variable, then it indeed yields an \textbf{optimal solution} to {\sc Single Shot}. Finally, we use Theorem \eqref{th1} to construct a $2$-approximation algorithm for {\sc Single Shot}. 
\subsubsection*{\hspace{-20pt} $\square$ LP-based Polynomial-time $2$-Approximation algorithm}
The (possible) non-integral optimal solution to {\sc RLP} may be converted to a feasible $2$-approximate optimal solution to {\sc Single Shot}. Let the value of the optimal solution to {\sc RLP} and the original {\sc Single Shot} problem be denoted by \textsf{OPT}\' ~and \textsf{OPT} respectively.  
It is obvious that any feasible solution to {\sc Single Shot} may be used to easily construct a feasible solution to {\sc RLP} with the same objective value. Hence, we readily have 
\begin{eqnarray} \label{lb1}
	\textsf{OPT}'\geq \textsf{OPT}.
\end{eqnarray}
%Also denote the objective value returned by the integral and the fractional portions of \textsf{OPT}' be $I,F_1,F_2$ respectively. 

Next, let the contribution to the total objective value \textsf{OPT}\' ~in Eqn. \eqref{rlp} by the standalone RBs ($x_{jk}=1$ for some $j$), and the (possible) one shared RB be $I,F'$ respectively, where $\textsf{OPT}'=I+F' $. The maximum contribution to the objective \eqref{opt} in the {\sc Single Shot} problem that we can obtain from any single RB, considering it separately, is $F_{\max}=\max_{j,k} \frac{1}{R_j}\min \{\gamma_{jk}, C\}$. Clearly, $F_{\max}\geq F'$. 
Finally, we choose the solution corresponding to the \emph{maximum} of $I$ and $F_{\max}$. It is clearly a feasible solution to {\sc Single Shot} and 
\begin{eqnarray}
	\max\{I, F_{\max}\} &\geq& \max\{I, F'\} \nonumber \\
	&\geq& \frac{1}{2} (I+F') \nonumber \\
	&=&\frac{1}{2}\textsf{OPT}'\geq \frac{1}{2}\textsf{OPT}. \label{approx2}
\end{eqnarray}
Eqn. \eqref{approx2} shows that the above LP-based scheme is a $2$-approximate poly-time algorithm to the {\sc Single Shot} problem with an overall PON capacity constraint. We summarize the above algorithm in Algorithm \ref{approx} below:

\begin{algorithm}
\caption{LP-based $2$-Approximation Algorithm for {\sc Single Shot}}
\label{approx}
\begin{algorithmic}[1]
\STATE Find the maximum possible objective value \eqref{opt} obtainable by using a \emph{single} RB, i.e., 
\begin{equation}
	F_{\max}=\max_{j,k} \frac{1}{R_j}\min \{\gamma_{jk}, C\}.
\end{equation}
\STATE Solve the Linear Program \textsc{RLP} \eqref{rlp}. Let $I$ be the objective value obtained by the standalone RBs (\emph{i.e.,} for which $x_{jk}=1$ for some $j$) in its optimal solution.   
\STATE Choose the solution corresponding to the maximum of $I$ and $F_{\max}$. 
\end{algorithmic}
\end{algorithm}

Although the algorithm {\sc Rounding Ad} has been shown to be a $2$-approximation algorithm, our numerical simulations in Section \ref{simulation} reveal that {\sc Rounding Ad} is likely to perform near-optimally in practice. In the following, we design a Dynamic Programming-based \emph{Optimal} algorithm  to {\sc Single Shot}. However, unlike the previous LP-based $2$-approximation algorithm, the Dynamic Program runs in pseudo-polynomial-time, and hence, it is less efficient than {\sc Rounding Ad}.  
\subsection{A DP-based Pseudo-Polynomial-time Optimal Algorithm}
%It is also possible to devise a pseudo polynomial-time algorithm based
%on Dynamic Programming (DP) for the problem when we impose the mild
%restriction 
%\footnote{This is not a restriction at all when all inputs are
%integrals (standard assumption). This follows from Lemma 3. }
Without any loss of generality, we may assume that the capacity $C$ and the wireless interface rates $\{\gamma_{jk}\}$ are integers. Then, Theorem \eqref{ad} readily implies that the optimal allocated rates $y_{jk}$'s are also integers for all $j,k$. Hence, without any loss of optimality, we may consider the following discrete range $\mathcal{R}_{jk}$ for the decision variable $y_{jk}$:
\begin{eqnarray} \label{range}
	\mathcal{R}_{jk} = \{0,1,2, \ldots, \gamma_{jk}\}, ~~\forall j,k.
	\end{eqnarray}
Let $\kappa$ denote the total number of RBs. Arrange the RBs in some order $(\textrm{RB}_1, \textrm{RB}_1, \ldots, \textrm{RB}_\kappa),$ and consider them adding to the solution one-by-one in this sequence. Let $V(M,k)$ denote the maximum objective value \eqref{opt} obtained by using only the first $k$ RBs with a PON of total capacity $M$. Then, as explained below, we have the following Dynamic Programming recursion for $V(M,k)$:
\begin{eqnarray}\label{dpr} 
	&& V(M, k)= \nonumber \\
	  && \hspace{-20pt} \max_j\max_{y_{jk} \in \mathcal{R}_{jk}, y_{jk} \leq M} \big(\frac{1}{R_j} y_{jk} + V(M-y_{jk}, k-1)\big), 
\end{eqnarray}
with $V(M,0)=0, \forall M$.

\emph{Optimality:} The above recursion \eqref{dpr} may be obtained as follows. Consider the $k$\textsuperscript{th} RB and suppose that it is assigned to the $j$\textsuperscript{th} user and allocated a rate of $y_{jk}$. Hence, $\textrm{RB}_k$ contributes a value of $\frac{y_{jk}}{R_j}$ towards the objective \eqref{opt}. With this assignment, we are left with the first $k-1$ RBs with a usable PON capacity of value $M-y_{jk}$, which, by the definition of $V(\cdot, \cdot)$, contributes a total value of $V(M-y_{jk}, k-1)$ to the objective \eqref{opt} in an optimal allocation. Hence, $V(M,k)$ is found by optimizing the total objective $\big(\frac{1}{R_j} y_{jk} + V(M-y_{jk}, k-1)\big)$ over the choice of user assignment $j$ and feasible rate allocation $y_{jk} \in \mathcal{R}_{jk}$ for the $k$\textsuperscript{th} RB. Moreover, since the total PON capacity is $M$, the variable $y_{jk}$ can only assume those values the feasible set $\mathcal{R}_{jk}$ which are at most $M$. This proves optimality of the DP recursion \eqref{dpr}. We summarize the DP algorithm below in Algorithm \ref{dp_code}.
\begin{algorithm}
\caption{Optimal Dynamic Program for {\sc Single Shot}}
\label{dp_code}
\begin{algorithmic}[1]
\STATE Set $V(M,0)\gets 0, ~~\forall M=0,1,2 \ldots, C$.
\STATE \FOR {$k=1,2, \ldots, \kappa$} 
\FOR {$M=0,1, \ldots, C$}
\STATE 
\begin{eqnarray*}
&& V(M, k)= \nonumber \\
	  && \hspace{-20pt} \max_j\max_{y_{jk} \in \mathcal{R}_{jk}, y_{jk} \leq M} \big(\frac{1}{R_j} y_{jk} + V(M-y_{jk}, k-1)\big).
	  \end{eqnarray*}
\ENDFOR
\ENDFOR
\STATE Return $V(C, \kappa)$.
\end{algorithmic}
\end{algorithm}

\section{Approximation Algorithm for the General {\bf Single Shot} Problem}
\label{s:algo}

In this Section, we present an approximation algorithm for the general {\sc Single Shot} problem, in which both the overall PON capacity constraint \eqref{eq:tot}, as well as the
individual RU-specific capacity constraints \eqref{eq:base} are active. 
Our proposed algorithm
is a greedy $2$-approximation algorithm called {\sc Matroid} which is based on the theory
of optimizing a sub-modular function over a partition matroid \cite{fisher1978analysis}. 
%The second algorithm is
%based on randomized rounding of the solution to an LP relaxation of the
%problem.  

\subsection*{Algorithm {\sc Matroid}}

%We first describe the Greedy algorithm {\sc matroid} for which the objective is always
%within a factor of $2$ of the optimal. 
For a feasible binary assignment vector $\bm{x}=(x_{ijk})$ with $\sum_j x_{ijk} \leq 1, \forall i,k$, define a corresponding set $S$ of RU-User-RB triples as follows:
\begin{equation*}
	S=\{ (i,j,k), ~~\textrm{if} ~~x_{ijk}=1\}.
\end{equation*}
Define the ground set $E=\{(i,j,k), \forall i,j,k\}$, and let $\mathcal{I}$ be the collection of all corresponding sets $S$ for all feasible binary assignment vectors $\bm{x}$. We make the following claim:

\begin{lemma} \label{partition_lemma}
The system $(E, \mathcal{I})$ is a {\bf partition matroid}. 	
\end{lemma}

%\begin{proof}
%	See Appendix \ref{partition_matroid}. 
%\end{proof}
 
 Moreover, for a given RB-to-User assignment $S \in \mathcal{I}$, we can efficiently compute the optimal rate assignments
$y_{ijk}$ values by using Algorithm~\ref{findy}. Let $f(S)$ be the
associated objective for the assignment $S$. 
We have the following lemma:
%It is not hard to show that $f(\cdot)$ is a non-decreasing submodular function. 
\begin{lemma} \label{submod_lemma}
	The set function $f(\cdot)$ is submodular.
\end{lemma}
\begin{proof}
See Appendix \ref{submod_proof}.	
\end{proof}

The greedy algorithm works by initializing 
$S$ to a null set and then repeatedly choosing a feasible augmentation $\bar{S}$ of $S$
 that maintains the feasibility constraint $\sum_j x_{ijk}\le 1$ and which
maximizes the increase in $f(S)$. The algorithm is summarized in Algorithm \ref{greedy}.

%In particular, for any vector $x$ let
%$n(x)$ be the number of $x_{ijk}$ variables in $x$ that are set to
%$1$. The algorithm works as follows:

\begin{algorithm}
\caption{Greedy Algorithm for {\sc Single Shot}}
\label{greedy}
\begin{algorithmic}[1]
\STATE $S\gets \phi$
\WHILE {1}
\STATE Find a feasible augmentation $\bar{S} \in \mathcal{I}$ of $S$ that maximizes $f(\bar{S})$ subject to the constraint
  $|\bar{S} \setminus S|=1$. 
  \IF {$f(\bar{S})=f(S)$}
  \STATE \textbf{break}
  \ELSE
  \STATE $S\gets \bar{S}$
  \ENDIF
  \ENDWHILE
\end{algorithmic}
\end{algorithm}

%\begin{itemize}
%\item Initialize $x$ to the zero vector.
%\item Repeat:
%\begin{itemize}
%\item Find a vector $\bar{x}$ that maximizes $f(\bar{x}-x)$ subject to
%  $n(\bar{x})-n(x)=1$. We can find such a $\bar{x}$ by considering
%  each possible $x_{ijk}(t)$ for augmenting $x$. 
%\end{itemize}
%\end{itemize}

\begin{lemma}
 Algorithm \ref{greedy} is a $2$-approximation algorithm for {\sc Single Shot}. 
\end{lemma}
\begin{IEEEproof}
This is a direct result of the Fisher-Nemhauser-Wolsey~\cite{fisher1978analysis} algorithm for
maximizing a submodular function over a matroid. 
\end{IEEEproof}
\emph{Discussion:} Although Algorithm \ref{greedy} also applies to the case when there is an overall PON capacity constraint (Section \ref{s:singleC}), the LP-based Algorithm \ref{approx} {\sc Rounding Ad} runs much faster than the Matroid-based Greedy Algorithm \ref{greedy}. This is because, each candidate augmentation in the  step 3 of Algorithm \ref{greedy} requires an invocation of Algorithm \ref{findy} for the evaluation of $f(\bar{S})$, which is costly.

\section{Heuristic Algorithms}
\label{s:prelim}
%\subsection{Heuristics}

In this Section, we present two natural heuristics that provide a baseline for numerical comparison with our proposed optimal and approximation algorithms in the following Section.\\~\\
$\bullet$ {\sc max-yield:} This is the simplest adaptation of the traditional 
PF algorithm so that it respects the capacity constraints. The algorithm works by 
going through the RUs and RBs in decreasing order of the index $ \max_j \gamma_{ijk}/R_{ij}$ sequentially and always picking the user
$j$ that maximizes the index $\gamma_{ijk}/R_{ij}$. 
At each step, the used and remaining capacity on the PON is tracked, and the algorithm 
stops when the available capacity is exhausted.\\~\\
%subject to there being available capacity on the PON. 
$\bullet$ {\sc max-value:} This algorithm works by going through
the RUs and RBs in decreasing order of the index $ \max_j \gamma_{ijk}/R_{ij}$ sequentially and always picking the user that maximizes $1/R_{ij}$.
At each step, the used and remaining capacity on the PON are tracked, and the algorithm stops when the available capacity is exhausted. 
Note that {\sc max-yield} tries to optimize the objective with respect to the wireless 
resources, and algorithm {\sc max-value} seeks to maximize the objective with respect to 
the PON capacity constraints.

\section{Simulation Results}\label{simulation}
\emph{Set-up:} We numerically simulate the performance of the proposed scheduling algorithms over a
service area of $1$ sq.\ km serving $1000$ wireless users via
$100$ RUs with overall one PON capacity constraint. We experiment with two different types of mid-hauls - one with a PON transport capacity of $C=1$ Gbps and another with a PON capacity of $C=1000$ Gbps. In the former case, the PON capacity is highly constraining, whereas in the latter case the PON capacity is hardly constraining. 
The users and the RUs are assumed to be distributed over the
service area according to a two-dimensional Poisson Point Process. The wireless channel has a bandwidth of $20$ MHz, and all RUs
have omnidirectional antennas with transmit power of $24$ \textrm{dBm}. The path-loss
coefficient is $\alpha_{\textrm{LOS}}=2.09$ for a line-of-sight transmit/receiver pair and
$\alpha_{\textrm{NLOS}}=3.75$ for a non-line-of-sight pair. The probability for a
transmit/receiver pair to be line-of-sight is $p_{\textrm{LOS}}=0.12$, if their separation
is less than $200$m and is zero otherwise. Each wireless channel has
a been simulated with a random fading process using Jakes' model \cite{wc1974microwave} with a maximum doppler shift of 
$10$ Hz. The wireless fading, in turn, determines the
air-interface rate-vector $\bm{\gamma}(t)$.\\

\begin{figure}[htb]
	\centering
		%\includegraphics[width=0.6\textwidth]{vRAN}
		%\begin{minipage}{0.49\textwidth}
		\hspace{-10pt}
		\begin{overpic}[width=0.35\textwidth]{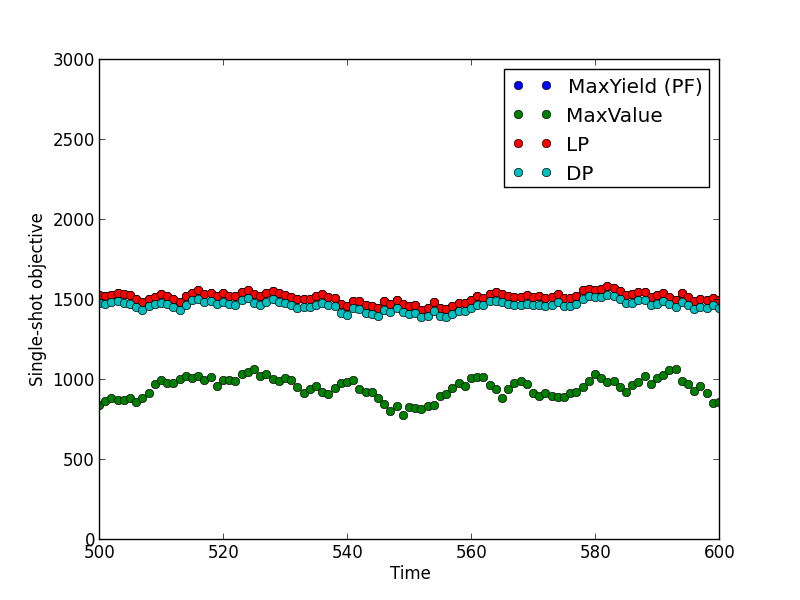}
		%\put(22,18){\footnotesize{(capacity bottleneck)}}
		%\put(70,35){\footnotesize{(time-varying air interface)}}
		%\end{overpic}	
		\end{overpic}
		\caption{{\sc Single Shot} objective ($C=1000$ Gbps)}
		\label{f:single-shot-1e12}
		\end{figure}
\begin{figure}[htb]
	\centering
		%\includegraphics[width=0.6\textwidth]{vRAN}
		%\begin{minipage}{0.49\textwidth}
		\hspace{-10pt}
		\begin{overpic}[width=0.35\textwidth]{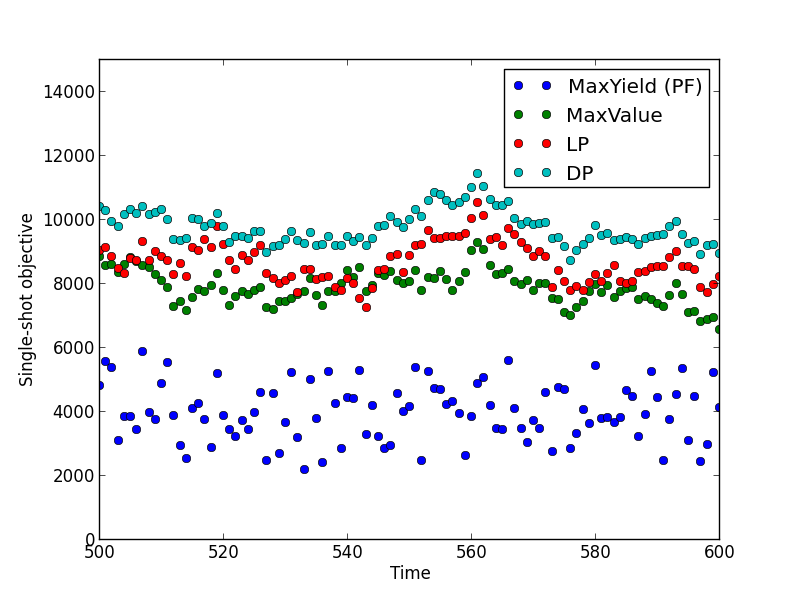}
		%\put(22,18){\footnotesize{(capacity bottleneck)}}
		%\put(70,35){\footnotesize{(time-varying air interface)}}
		%\end{overpic}	
		\end{overpic}
		\caption{{\sc Single Shot} objective ($C=1$ Gbps)}
		\label{f:single-shot-1e9}
		\end{figure}
\begin{figure}[htb]
	\centering
		%\includegraphics[width=0.6\textwidth]{vRAN}
		%\begin{minipage}{0.49\textwidth}
		\hspace{-10pt}
		\begin{overpic}[width=0.35\textwidth]{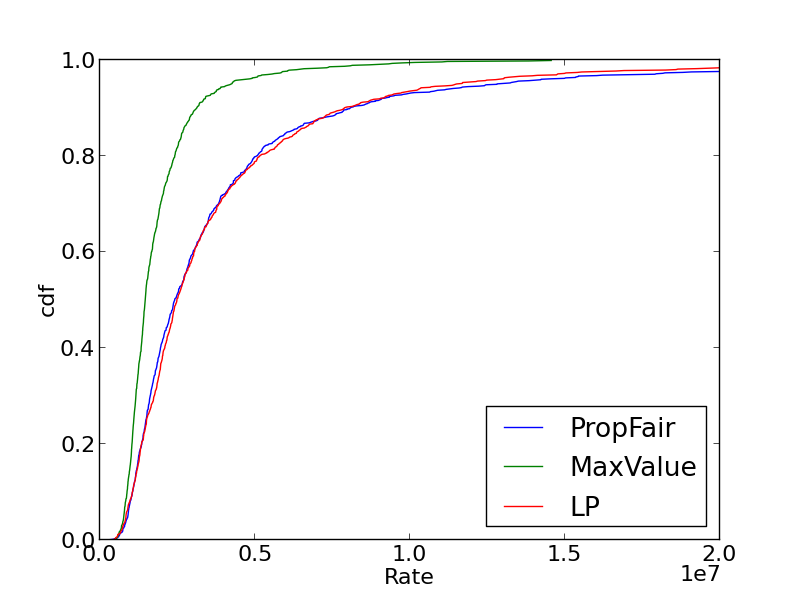}
		%\put(22,18){\footnotesize{(capacity bottleneck)}}
		%\put(70,35){\footnotesize{(time-varying air interface)}}
		%\end{overpic}	
		\end{overpic}
		\caption{Long-term user rate distribution ($C=1000$ Gbps)}
		\label{f:multi-shot-1e12}
		\end{figure}
\begin{figure}[htb]
	\centering
		%\includegraphics[width=0.6\textwidth]{vRAN}
		%\begin{minipage}{0.49\textwidth}
		\hspace{-10pt}
		\begin{overpic}[width=0.38\textwidth]{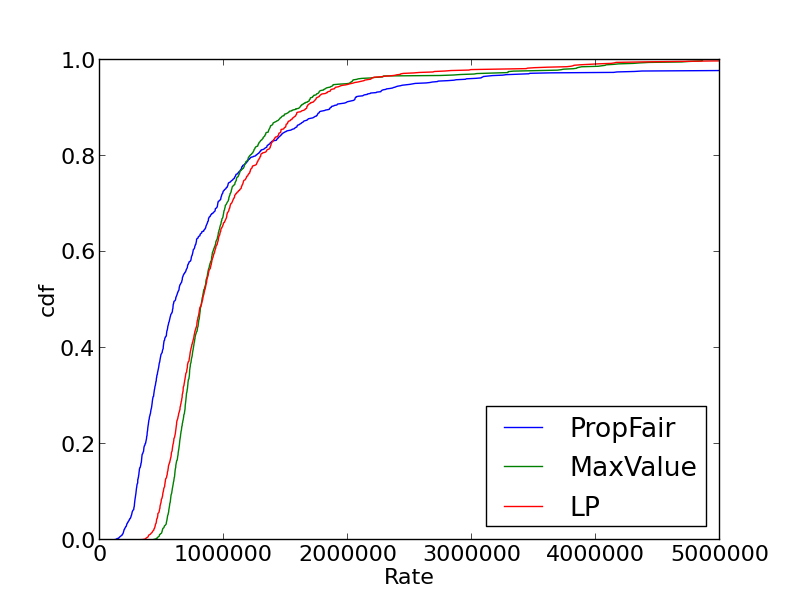}
		%\put(22,18){\footnotesize{(capacity bottleneck)}}
		%\put(70,35){\footnotesize{(time-varying air interface)}}
		%\end{overpic}	
		\end{overpic}
		\caption{Long-term user rate distribution ($C=1$ Gbps)}
		\label{f:multi-shot-1e9}
		\end{figure}

\emph{Results and Discussion:} In Figures~\ref{f:single-shot-1e12} and \ref{f:single-shot-1e9}, we
show the {\sc Single Shot} objective achieved by {\sc max-yield}, {\sc max-value}, the Dynamic Programming
algorithm DP, and the LP-based $2$-approximation algorithm {\sc Rounding-Ad}. We plot
these values over 100 time slots (after a warm-up period). To fairly compare the efficacy of the above algorithms one a slot-by-slot basis.  In our simulations, we assume that all algorithms use the same set of $R_{ij}(t)$ values (that are commonly driven by the {\sc max-yield}
algorithm). 

For the case $C=1000$ Gbps, the PON capacity is not constraining and hence, the
{\sc max-yield} algorithm is optimal. In this case, we observe from
Figure~\ref{f:single-shot-1e12} that the DP and LP-based
approaches are essentially optimal as well. (The points for LP
coincide with (and therefore hide) the points for {\sc max-yield}.) On the other hand, Figure \ref{f:single-shot-1e9} shows that for the case of $C=1$ Gbps, a pure
Proportional Fair approach would violate the PON capacity constraint, and hence,
{\sc max-yield} is not optimal. Moreover, {\sc max-value} is also not
optimal either since it would not necessarily fill up the PON capacity. We see
that both the DP and LP-based algorithms perform better than the heuristics.
The fact that DP and LP-based algorithms work well
in both cases illustrates the benefits of scheduling with an awareness
of both the channel conditions and the PON capacity. 

In Figures~\ref{f:multi-shot-1e12} and \ref{f:multi-shot-1e9} we
observe the overall user rate distribution when we run Proportional
Fair ({\sc
  max-yield}), {\sc max-value} and the LP-based $2$-approximation algorithm {\sc Rounding Ad}. If $C=1000$ Gbps,
then {\sc Rounding Ad} has a similar rate distribution to the optimal
Proportional Fair algorithm. If $C=1$ Gbps, then {\sc Rounding Ad} outperforms
Proportional Fair ({\sc
  max-yield}) and {\sc max-value} algorithms for all except the users with the
highest channel rates. This, in turn, means that {\sc Rounding Ad} leads
to a higher value of the logarithmic utility function, which rewards
fairness. 

\iffalse
The plot in Figure \ref{per-slot-wt} gives a direct per-slot
comparison of achieved weights in the max-weight algorithms for three
scheduling schemes. The state-update \eqref{state-update} of the
max-weight algorithm is driven by the Proportional Fair algorithm. As
a result, this plot shows the effectiveness of the per-slot max-weight
algorithms proposed. It follows that the traditional proportional fair
algorithm, which does not take into account the midhaul capacity
constraint, performs very poorly (more than 2X) compared to the
optimal and the heuristic {\sc max-value} algorithm regarding the
{\sc single-shot} objective.  The actual user-rate distribution when
driven separately by three algorithms has been plotted in Figure
\ref{cdf-dist}. It is evident from the plot that the {\sc Max-Value}
and the optimal {\sc DP} algorithm out-performs the greedy
proportional fair algorithm in terms of fairness and aggregate user
rates.
\fi

       % \end{minipage}
%\begin{minipage}{0.49\textwidth}
%\vspace{-10pt}
%\hspace{10pt}
%\begin{figure}
%\begin{overpic}[width=0.45\textwidth]{rate-dist}
%\end{overpic}
%\end{minipage}
		%\caption{Schematic of a \textcolor{blue}{vRAN} architecture}
%		\caption{\small{CDF of long-term user rate-distributions }}
		\label{cdf-dist}
%\end{figure}

\section{Discussion and Related Work} \label{discussions}
In this paper, we have analyzed a scheduling problem that arises in the context of virtualized RAN architecture with a fixed capacity PON mid-haul. The importance of this type of scheduling problem is on the rise given the shift towards more flexible split-processing architecture for 5G wireless networks. %This type of problem will arise with
%increasing frequency as 5G wireless networks migrate to a split
%processing architecture~\cite{DotschDMSSS13}. 
We view our work as a natural extension of the large body literature on
scheduling over time-varying channels \cite{agrawal2002optimality,andrews2005optimal,borst2001dynamic,borst2003user,liu2001opportunistic,liu2003framework,shakkottai2001scheduling,shakkottai2002scheduling,tassiulas1992stability}. This body of
research introduced and influenced the popular Proportional Fair (PF) algorithm \cite{viswanath2002opportunistic,jalali2000data}, which is
implemented in almost all of today's cellular network.  Our
algorithms lead to a different and more efficient scheduling than the usual PF algorithm when the
PON mid-haul capacity is the bottleneck. 

%For future work, we are therefore interested in whether there
%exist simpler approaches that are still able to balance the competing
%demands of the air interface and the PON mid-haul. 

\bibliographystyle{IEEEtran}
\bibliography{pon}

\appendix 
\subsection{Proof of Lemma \ref{l:structure}}
\begin{IEEEproof} \label{structure:proof}
Let $\bm{y}^*$ denote the optimal solution for the assignment $\bm{x}$. Consider the $(ijk)$ triples in the order above and consider the first triple $(ijk)$
for which the $y_{ijk}(t)$ value according to the above algorithm is different
from $y^*_{ijk}(t)$.  Since the $y_{ijk}(t)$ have been made as large
as possible subject to all of the constraints, it must be the case
that $y^*_{ijk}(t)<y_{ijk}(t)$. We now increase in a continuous manner
until $y^*_{ijk}(t)=y_{ijk}(t)$. In order to do this, we might have to
decrease some other $y^*$ values.  It is not necessary to do this for
$i'j'k'$ triples that have already been considered since $y_{ijk}(t)$
does not violate any constraints.  If there is a value $y^*_{ij'k'}(t)>0$ for
a later triple $(ij'k')$ at the same RU then we decrease it until either it hits zero
or all the constraints are satisfied.  If there is no such
$y^*_{ij'k'}(t)$ at the same RU then we do the same but for a
later $y^*_{i'j'k'}$ value at a different RU. We can always
find such a value since otherwise $y^*_{i'j'k'}(t)=0$ for all later
triples. This cannot be true if the $y^*$ values satisfy the
constraints since the $y$ values represent a feasible solution. 

Since we are decreasing $y^*_{i'j'k'}(t)$ for a later triple it must
be the case that $1/R_{ijk}(t)\ge 1/R_{i'j'k'}$. Hence the objective
function for the $y^*$ values cannot get any worse as we make the
changes. If we keep repeating the procedure then eventually the $y^*$
values will equal the $y$ values. This implies that the $y$ values
found from the above procedure give us an optimal solution.
\end{IEEEproof}
%Lemma~\ref{l:structure} allows us to define a function $f({\bf x})$ that
%equals the optimal value of the objective for any given vector ${\bf
%  x}$ whose components consist of the $x_{ijk}(t)$ values. 

\subsection{Alternative Proof of Theorem \ref{ad}} \label{alter}
\begin{IEEEproof}
Let $\bm{x}^*$ be an optimal RB assignment to the {\sc Single Shot} problem \eqref{opt}. We show that there exists an \textsc{Almost Discrete} feasible optimal allocation $\bm{y}^*$.	\\
 For each RB $k$, there exists exactly one user $j^*(k)$ such that $x^*_{j^*(k)k}=1$.  Denote $R_{j^*(k)k}=R_k$ and $\gamma_{j^*(k)k}=\gamma_k$. Then substituting this optimal $\bm{x^*}$ in \eqref{opt}, we note that the optimal allocation $\bm{y}$ is a solution of the following LP: 

\begin{eqnarray} \label{LPn}
	\max \sum_k y_k/R_k
\end{eqnarray}

s.t., 
\begin{eqnarray}
	&&\sum_k y_k \leq C \label{simplex1}\\
	&& 0 \leq y_k \leq \gamma_k, \hspace{10pt} \forall k.  \label{bd2}
\end{eqnarray}
It is easy to see that the feasible region of the LP, given by the Eqns \eqref{simplex1} and \eqref{bd2}, is bounded, and hence, a finite optimal solution to the LP \eqref{LPn} exists. 
Next, recall the fundamental result that the solution of an LP is always obtained at a vertex of the feasible region \cite{papadimitriou1998combinatorial}, which is also known as the \textbf{B}asic \textbf{F}easible \textbf{S}olution (BFS). \\
Let $\kappa$ be the number of RBs. Since there are $(\kappa+1)$ linearly independent inequalities in the constraints
\eqref{simplex1} and \eqref{bd2}, and the dimension of the vector
$\bm{y}$ is $\kappa$, it is clear that at least $(\kappa-1)$ of inequalities from \eqref{bd2} must be active in any {\bf BFS}. This proves
that there exists an optimal solution which is \textsc{Almost
  Discrete}. 

\end{IEEEproof}
\tcb{
\subsection{Proof of Theorem \ref{d2ad}}
\begin{IEEEproof}
From Theorem~\ref{ad}, we know that there is an optimal solution to
\eqref{opt} that is \textsc{Almost Discrete}. In such a solution, call a RB $k$ \textsf{Utilized} if $y_{jk}^*=\gamma_{jk}$, and \textsf{Under-Utilized} if $0< y_{jk} < \gamma_{jk}$ for some $j$ with $x^*_{jk}=1$. 
Then we have, 
\begin{eqnarray*}
%\begin{center}	
%\end{eqnarray*}
\textsf{OPT}&=& \textsf{Utilized} \textrm{~RBs} + \textsf{Under-Utilized} \textrm{~RB}\\
&\leq & 2 \max \{ \textsf{Utilized} \textrm{~RBs},\textsf{Under-Utilized} \textrm{~RB}\} 
%\end{center}
\end{eqnarray*}
Hence,
\begin{eqnarray} \label{fundamental_ineq}
	\max \{ \textsf{Utilized} \textrm{~RBs},\textsf{Under-Utilized} \textrm{~RB}\} \geq \frac{1}{2} \textsf{OPT}
\end{eqnarray}
Now consider optimizing each of the terms appearing on the LHS of \eqref{fundamental_ineq} \emph{separately}, and choosing the better of these two solutions. 

Maximizing the second term is easy, we just take the maximum over the
various RBs and allocate it to the full extent, i.e.,
\begin{eqnarray}
	k^*= \arg \max _{jk} \frac{1}{R_j} \min\{\gamma_{jk}, C\}.
\end{eqnarray}
An optimal solution to \textsc{Discrete} maximizes the first term. The proof now follows from Eqn. \eqref{fundamental_ineq}.
\end{IEEEproof}
}

\subsection{Proof of Lemma \ref{partition_lemma}} \label{partition_matroid}
\begin{proof}
Define the following disjoint partition of the base set $E=\bigcupdot_{i,k} E_{ik}$, where
\begin{equation*}
	E_{ik}=\{(i,j,k),\forall j\}.
\end{equation*}
	Since $\bm{x}$ is feasible, it follows that at most one user may be assigned to any RB at each
RU. Thus, for any $S \in \mathcal{I}$, we have $|S\cap E_{ik}|\leq 1, \forall i,k.$ Hence, the proof directly follows from Example 12.8 (p. 288) of \cite{papadimitriou1998combinatorial}. 
\end{proof}

\subsection{Proof of Lemma \ref{submod_lemma}}\label{submod_proof}
\begin{IEEEproof}
Let $S$ be a set of RBs and consider two RBs $j, k$ such that $j \in S^c$ and $k \in S^c$. Then, following \cite{fisher1978analysis}, the following inequality \eqref{submod_def} establishes submodularity of the function $f(S)$:
\begin{equation} \label{submod_def}
	f(S\cup \{j\})+ f(S\cup \{k\}) \geq f(S\cup \{j,k\})+ f(S).
\end{equation}
To show the above inequality, we exploit the result in Lemma \ref{l:structure}. Suppose, in the optimal solution corresponding to $f(S\cup j)$, the RB $j$ was allocated a rate of $r_j$. Similarly, in the optimal solution to $f(S\cup k)$, the RB $k$ was allocated a rate of $x_k$. And finally, in the optimal solution corresponding to $f(S\cup \{j,k\})$, the RBs $j$ and $k$ was allocated a rate of $y_j$ and $y_k$ respectively.  Then, it is clear from the Lemma \ref{l:structure} that $r_j \geq y_j$ and $x_k \geq y_k$. Then, we may write
\begin{eqnarray*}
f(S \cup \{j\})-f(S)&=& r_j \delta_j\\
f(S \cup \{k\})-f(S) &=& x_k \delta_k\\
f(S\cup \{j,k\}) - f(S)&=& y_j \delta_j + y_k \delta_k,	
\end{eqnarray*}
where $\delta_j $ and $\delta_k$ denote the marginal utility of adding the RBs to the set $S$ (this situation is equivalent to adding a new variable in the simplex method for solving the LP and the quantities $\delta_j$ and $\delta_k$ correspond to the reduced costs of the variables corresponding to the RBs $j$ and $k$ \cite{bertsimas1997introduction}). Using the above relations, we have 
\begin{eqnarray*}
&&\big(f(S \cup \{j\})-f(S)\big) + \big(f(S \cup \{k\})-f(S)\big) \\
&\geq & f(S\cup \{j,k\}) - f(S).
\end{eqnarray*}
Rearranging the above, we have 
\begin{eqnarray*}
	f(S\cup \{j\})+ f(S\cup \{k\}) \geq f(S\cup \{j,k\})+ f(S),
\end{eqnarray*}
and this establishes the submodularity property of the function $f(\cdot)$. 
\end{IEEEproof}

\end{document}